\documentclass{article}
\usepackage{theorem,graphicx,amssymb,amsmath}
\usepackage[retainorgcmds]{IEEEtrantools}
\usepackage{lineno}

\theorembodyfont{\slshape}

\newtheorem{theorem}{Theorem}
\newtheorem{lemma}[theorem]{Lemma}

\newtheorem{defini}{Definition}

\def\QED{\ensuremath{{\square}}}

\def\markatright#1{\leavevmode\unskip\nobreak\quad\hspace*{\fill}{#1}}
\newenvironment{proof}
  {\begin{trivlist}\item[\hskip\labelsep{\bf Proof.}]}
  {\markatright{\QED}\end{trivlist}}

\newcommand{\conv}{\operatorname{Conv}}
\newcommand{\vol}{\operatorname{Area}}

\title{Balanced Islands in Two Colored Point Sets in the Plane\footnote{
Oswin Aichholzer and Birgit Vogtenhuber were supported by the ESF EUROCORES programme EuroGIGA--ComPoSe, Austrian Science Fund (FWF): I 648-N18.
Pablo Perez-Lantero was partially supported by projects CONICYT FONDECYT/Iniciaci\'on 11110069 (Chile), and Millennium Nucleus Information and Coordination in Networks ICM/FIC RC130003 (Chile).}}

\author{Oswin Aichholzer\thanks{Institute of Software Technology, Graz University of Technology, Graz, Austria.\texttt{[oaich|bvogt]@ist.tugraz.at}}
\and Nieves Atienza\thanks{Departamento de Matem\'atica Aplicada, Universidad de Sevilla, Spain. \texttt{natienza@us.es}}
\and Jos\'e M. D\'iaz-B\'a\~nez\thanks{Departamento de Matem\'atica Aplicada II, Universidad de Sevilla, Spain. \texttt{dbanez@us.es}}
\and Ruy Fabila-Monroy\thanks{Departamento de Matem\'aticas, Cinvestav, D.F., M\'exico. \texttt{ruyfabila@math.cinvestav.edu.mx}}
\and David Flores-Pe\~naloza\thanks{Departamento de Matem\'aticas, Facultad de Ciencias, UNAM, Mexico. \texttt{dflorespenaloza@gmail.com}}
\and Pablo P\'erez-Lantero\thanks{Departamento de Matem\'atica y Ciencia de
la Computaci\'on, Universidad de Santiago, Santiago, Chile. \texttt{pablo.perez.l@usach.cl}}
\and Birgit Vogtenhuber\footnotemark[2]
\and Jorge Urrutia\thanks{Instituto de Matem\'aticas, UNAM, M\'exico. \texttt{urrutia@matem.unam.mx}}}

\begin{document}
%\linenumbers
\maketitle

\begin{abstract}
Let $S$ be a set of $n$ points in general position in the plane, $r$ of which are red and $b$ of which are blue.
In this paper we prove that there exist:
for every $\alpha \in \left [ 0,\frac{1}{2} \right ]$, a convex set containing exactly $\lceil \alpha r\rceil$ red points
and exactly $\lceil \alpha b \rceil$ blue points of $S$; a convex set containing exactly $\left \lceil \frac{r+1}{2}\right \rceil$ red points
and exactly $\left \lceil \frac{b+1}{2}\right \rceil$ blue points of $S$.
Furthermore, we present polynomial time algorithms to find these convex sets.
In the first case we provide an $O(n^4)$ time algorithm  and
an $O(n^2\log n)$ time algorithm in the second case.
Finally, if $\lceil \alpha r\rceil+\lceil \alpha b\rceil$ is small, that is, not much larger
than $\frac{1}{3}n$, we improve the running time to $O(n \log n)$.
\end{abstract}

\newpage
%\tableofcontents

\section{Introduction}

Let $S$ be a set of $n$ points in the plane, $r$ of which are red and $b$ of which are blue. 
Without loss of generality, we assume that $S$ (and any other finite point set in this paper) is in general position, that is, no three points lie on a common line. 
A large class of problems in Discrete and Computational Geometry involves partitioning such point sets.
A typical question in this context is whether a given 2-colored point set may be partitioned into parts that satisfy certain predefined properties.
In this paper, we present algorithms for computing convex sets that contain a balanced proportion of points of $S$ of each color.

The Ham Sandwich theorem states that there exists a straight line that simultaneously partitions the red points and the blue points in half. 
As a consequence, there exists a convex set  containing half of the red points and half of the blue points of $S$.
This result can be generalized as follows.

\begin{theorem}\label{thm:hobby}\textbf{(The Balanced Island Theorem)}
Let $S$ be a set of $r$ red points and $b$ blue points in the plane.
Then, for every  $\alpha \in \left [ 0,\frac{1}{2} \right ]$ there exists:
\begin{enumerate}
  \item a convex set containing exactly $\lceil \alpha r\rceil$ red points and exactly $\lceil \alpha b \rceil$ blue points of $S$;
  \item a convex set containing exactly $\left \lceil \frac{r+1}{2}\right \rceil$ red points and exactly $\left \lceil \frac{b+1}{2}\right \rceil$ blue points of $S$.
\end{enumerate}
\end{theorem}

An \emph{island} of $S$ is a subset $I$ of $S$ such that $\conv(I)\cap S =I$.
The first case of Theorem~\ref{thm:hobby} implies in particular, that if $r=b$ then
for every $k=1,\dots, \lceil \frac{r}{2} \rceil$ there exists an island containing 
$k$ red points and $k$ blue points of $S$; such an island contains a balanced
number of red and blue points.  See Figure~\ref{fig:example} for an example.

On the other hand, consider the following construction.
Place $r$ red points at the vertices of a regular $r$-gon and place
$b$ blue points close to its center. Every convex set containing 
$\left \lceil \frac{r+1}{2}\right \rceil+1$ red points contains all the blue points.
The gap between $\left \lceil \frac{r}{2}\right \rceil$ and $\left \lceil \frac{r+1}{2}\right \rceil+1$
is filled by the second case of Theorem~\ref{thm:hobby}.

\begin{figure}
  \begin{center}
   \includegraphics[width=0.5\textwidth]{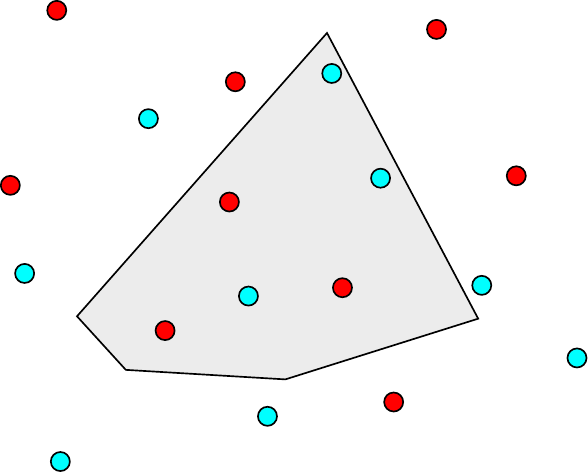}
\end{center}
\caption{An example of a balanced island for $r=b=9$ and $\alpha=1/3$.}
\label{fig:example}
\end{figure}

In higher dimensions the Ham Sandwich theorem states that the color 
classes of a $d$-colored finite set of points  in $\mathbb{R}^d$ can be 
simultaneously bisected by a hyperplane. Although the Ham Sandwich 
theorem can be easily proven in dimension two, already
for dimension three, tools from Algebraic Topology 
are needed. This is also the case for Theorem~\ref{thm:hobby}.
The first case of Theorem~\ref{thm:hobby} can be derived from
a theorem of Blagojevi\'c and Dimitrijevi\'c Blagojevi\'c (Theorem 3.2 in \cite{equivariant}), 
which is a prime example of the applications of Algebraic Topology in Combinatorial Geometry.
This implication was also noted by Sober\'on as a remark in his PhD thesis~\cite{pablo}. 

The connection between Algebraic Topology and Combinatorial Geometry
can sometimes be hard to understand without the proper background.
For the sake of self-containment, 
we include an expository account of this connection in Section~\ref{sec:topo}.

We prove Theorem~\ref{thm:hobby} in Section~\ref{sec:hobby}. In Lemma~\ref{lem:gen}
we show how to derive the first case of Theorem~\ref{thm:hobby} from the result
of \cite{equivariant};  the second case of Theorem~\ref{thm:hobby} needs a separate proof, 
which we give in Lemma~\ref{lem:n+1}; the argument used is also topological in essence.

Finally,  in Section~\ref{sec:algo}, we consider the algorithmic facet of Theorem~\ref{thm:hobby}. In 
 Theorem~\ref{thm:hobby_alg} we show that the convex set guaranteed by Theorem~\ref{thm:hobby} can
be found in $O(n^4)$ time in the first case, and in  $O(n^2\log n)$ time
in the second case. We also show in Theorem~\ref{thm:alg_2} that if 
$\lceil \alpha r\rceil+\lceil \alpha b \rceil$ is small, that is, not much larger than
$\frac{1}{3} n$, the running time  can be improved
to $O(n \log n)$.

\section{Topological preliminaries}\label{sec:topo}

\subsection{The Ham Sandwich and Borsuk-Ulam theorems}\label{sec:ham}

The statement that there exists a straight line simultaneously
bisecting the color classes of a two-colored finite point
set in the plane, is what many computational geometers would
recognize as the Ham Sandwich theorem. It generalizes
to higher dimensions as follows.

\begin{theorem} \textbf{(Discrete Ham Sandwich theorem).} \label{thm:ham_d}

Let $S_1, \dots, S_d$ be finite point sets in $\mathbb{R}^d$.
Then there exists a hyperplane that simultaneously bisects\footnote{Each of the two open half-spaces defined
by the hyperplane contain at most $\left \lfloor |S_i|/2 \right \rfloor$ points
of  $S_i$. } each $S_i$.
\end{theorem}

At first sight it might be hard to see the connection
of Theorem~\ref{thm:ham_d} with Topology. This connection
perhaps is more apparent in the following continuous version of Theorem~\ref{thm:ham_d}.

\begin{theorem} \textbf{(Continuous Ham Sandwich theorem).} \label{thm:ham_c}

Let $U_1, \dots, U_d$ be bounded open sets in $\mathbb{R}^d$.
Then there exists a hyperplane that simultaneously bisects\footnote{Each of the two open half-spaces defined
by the hyperplane contain half of the volume of $U_i$.} each $U_i$.
\end{theorem}

We point out that Theorem~\ref{thm:ham_c} is usually
stated in the more general setting of finite Borel measures. To keep our 
exposition as self-contained as possible, we opted to use volumes
of open sets instead. 

The discrete version of the Ham Sandwich theorem can be proven using the 
continuous version. Given $S_1, \dots, S_d$ 
finite point sets in $\mathbb{R}^d$, the first step
is to replace each point with a ball of radius $\varepsilon >0$
centered at the point. The continuous Ham Sandwich theorem ensures that 
there exists a hyperplane that simultaneously bisects
these expanded $S_i$'s. If we let $\varepsilon$ tend to zero this hyperplane
converges to a hyperplane that simultaneously bisects the original $S_i$'s.

The continuous version of the Ham Sandwich theorem can be
proven using the Borsuk-Ulam theorem. The Borsuk-Ulam theorem has
many equivalent formulations. One of them 
states that for any map (continuous function) from the  $d$-dimensional sphere $S^d$
to $\mathbb{R}^d$, there exists a pair of antipodal points
with the same image. 

\begin{theorem} \textbf{(Borsuk-Ulam theorem A).} \label{thm:borsuk_p}

For every map $f:S^d \to \mathbb{R}^d$ there
exists a point $x \in S^d$ such that $f(x)=f(-x)$.
\end{theorem}

To prove the continuous version of the Ham Sandwich theorem, one first chooses a set
of a given family of bounded open sets $U_1,\dots, U_d$ of $\mathbb{R}^d$.
Say $U_d$ is chosen. For every possible direction $\vec{v} \in S^{d-1}$, consider the set of oriented
hyperplanes, orthogonal to $\vec{v}$, that bisect $U_d$. These hyperplanes form an interval
along this direction. Let $\Pi_{\vec{v}}$ be the first such hyperplane. 

It can be checked that the set $\{\Pi_{\vec{v}}:\vec{v} \in S^{d-1}\}$
is topologically equivalent (homeomorphic) to $S^{d-1}$. In this setting, 
pairs of antipodal points
correspond to pairs of oriented planes with parallel supporting planes and 
with opposite orientation, such that the volume of $U_d$ that
is contained between them is equal to zero. A map $f$ from this space 
of oriented planes to $\mathbb{R}^{d-1}$, is defined by 
mapping each such plane $\Pi$ to the point $f(\Pi) \in \mathbb{R}^{d-1}$ whose
$i$-th coordinate is the fraction of the volume of $U_i$ that lies above $\Pi$.
By the Borsuk-Ulam theorem, there exists a plane $\Pi$ that has
the same fraction of the volume of $U_i$ above it as its antipodal plane $-\Pi$ has. 
Therefore, $\Pi$ and $-\Pi$ simultaneously
bisect every $U_i$.

For more applications of the Borsuk-Ulam theorem see Matou{\v{s}}ek's book~\cite{borsuk}.

\subsection{Equivariant maps}\label{sec:eq_map}

The Borsuk-Ulam theorem as stated in Theorem~\ref{thm:borsuk_p} is
formulated in a positive way---it ensures the existence of a pair
of antipodal points of $S^d$ with a certain property.
It can also be formulated in a negative way---that no map 
from $S^d$ to $S^{d-1}$ with a certain property (antipodality) exists. A continuous function 
$f:S^{d} \to S^{d-1}$ is \emph{antipodal} if $f(-x)=-f(x)$ for all
$x \in S^d$. We explicitly give this negative formulation of the Borsuk-Ulam theorem.

\begin{theorem}  \textbf{(Borsuk-Ulam theorem B).} \label{thm:borsuk_n}

There is no antipodal map from $S^{d}$ to $S^{d-1}$.
\end{theorem}

Antipodality and the negative formulation of the Borsuk-Ulam theorem
are examples of a more general phenomena. In the case of antipodality,
 consider the map that sends every point $x \in S^d$
to its antipodal point. This map together with the identity on $S^d$
form a group under function composition. This group is isomorphic to the group $\mathbb{Z}_2$ 
(the unique group with two elements).
So $\mathbb{Z}_2$ is said to \emph{act} on $S^d$. The formal and more
general definition is the following.

\begin{defini}
An \textbf{action} of a group $G$ on a topological space $X$ is 
an assignment of an homeomorphism $\phi_g$ of $X$ to every
element $g$ of $G$, such that

\begin{itemize}

\item $\phi_e$ is the identity on $X$ if and only if $e$ is the identity element of $G$.

\item $\phi_g\phi_h=\phi_{gh}$ for all $g,h \in G$.

\end{itemize} 

\end{defini}

The \emph{action} of an element $g \in G$ on a point $x \in X$ is defined as the point
$\phi_g(x)$. If the action is already specified or implied,
we simply write $gx$. We say that $\phi$  is \emph{free}
if $\phi_g(x)=x$ for some $x \in X$ implies that $g=e$.
In other words, the only homeomorphism that
maps some point to itself is the one assigned
to the identity element. 

For a given group $G$ acting on two topological
spaces $X$ and $Y,$ a $G$-\emph{equivariant} map
is a map $f:X \to Y$
such that $f(gx)=gf(x)$ for all $x \in X$ and $g \in G$.
An antipodal map from $S^d$ to $S^{d-1}$ is just a 
$\mathbb{Z}_2$-equivariant map. The negative formulation of the
Borsuk-Ulam theorem can be reinterpreted as the statement that 
there is no $\mathbb{Z}_2$-equivariant map
from $S^d$ to $S^{d-1}$, where $\mathbb{Z}_2$ acts
freely on both $S^{d}$ and $S^{d-1}$.

The negative formulation of the Borsuk-Ulam theorem can
be used to prove the Ham Sandwich theorem in the following way.
Assume that there exists a family of bounded open sets $U_1,\dots,U_{d}$
of $\mathbb{R}^d$ for which there is no hyperplane that simultaneously
bisects all of them. To apply the negative version
of the Borsuk-Ulam theorem, we use this assumption to define
an antipodal ($\mathbb{Z}_2$-equivariant) map 
from $S^{d-1}$ to $S^{d-2}$; thus arriving to a contradiction.

We proceed in a similar way as when using the positive
formulation of the Borsuk-Ulam theorem. Again, we use the map 
$f$ from the set of oriented planes that bisect $U_d$ to $\mathbb{R}^{d-1}$.
(Recall that $f(\Pi)$ is the point of $\mathbb{R}^{d-1}$ whose 
$i$-th coordinate is the fraction of the volume of $U_i$ that lies above
the hyperplane $\Pi$.)
The image of $f$ is actually the $(d-1)$-dimensional cube $[0,1]^{d-1}$ rather
than all of $\mathbb{R}^d$; we will regard
$f$ as a map from $S^{d-1}$ to $[0,1]^{d-1}$. 

The topological space $[0,1]^{d-1}$ can be equipped with an ``antipodal''
function by mapping every point $(x_1,\dots, x_{d-1}) \in [0,1]^{d-1}$ to the point 
whose $i$-th coordinate is $1-x_i$. 
The assumption that there is no hyperplane simultaneously
bisecting all the $U_i$'s, is equivalent to the assumption
that no bisecting plane of $U_{d}$ is mapped to the point 
$p:=\left ( \frac{1}{2}, \frac{1}{2},\dots, \frac{1}{2} \right )$.
Actually, $f$ is a $\mathbb{Z}_2$-equivariant map from
$S^{d-1}$ to $[0,1]^{d-1} \setminus \{ p \}$. 

Let $x$ be a point
in $[0,1]^{d-1} \setminus \{ p \}$ and let $\gamma_x$ be the
infinite ray with apex $p$ that passes through $x$. 
Let $g$ be the map that sends $x$ to the intersection of $\gamma_x$ and
the boundary of $[0,1]^{d-1}$. The boundary of $[0,1]^{d-1}$ is homeomorphic
to $S^{d-2}$ and the antipodal function on $[0,1]^{d-1}$
defines a free $\mathbb{Z}_2$-action when restricted to it.
The function $g \circ f$ is the desired 
antipodal map from $S^{d-1}$ to $S^{d-2}$. 

The method just described to prove the Ham Sandwich
theorem is certainly more involved than the method described in 
Section~\ref{sec:ham}. However, it is this approach that has been streamlined
to prove many equipartition theorems in what has been
called the ``Configuration Space-Test Map'' (CS-TM) scheme. 
For a nice survey of the CS-TM scheme 
see {\v{Z}}ivaljevi{\'c}'s paper~\cite{user_guide}.

\subsection{Configuration Space-Test Map Scheme}

Suppose that we want to prove that an object with a certain
 property exists. (In the
previous example we searched for a hyperplane that simultaneously
bisects a given family $U_1,\dots, U_d $ of bounded
open sets of $\mathbb{R}^{d}$.)
The approach of the CS-TM scheme
is as follows.

A set of candidates for the solution is first defined. 
This set is then given a topology and a free action of a group $G$.
This space $X$ is called the \emph{configuration space}. 
In our previous example, $X$ was the space of oriented hyperplanes
that bisect $U_d$.

Afterwards, a map $f$ from $X$ to a space $Y$ is defined.  
The desired object is then shown to exist
if and only if some point of $X$ is mapped by $f$ to a certain
subspace $Z$ of $Y$. The space $Y$ and its subspace
$Z$ are called the \emph{test space}; the map $f$ is called
the \emph{test map}. In our previous examples
the roles of $Y, Z$ and $f$ where played by the boundary of $[0,1]^{d-1}, \{ p \}$, 
and $f$, respectively.

An action of $G$ on $Y$ is defined
so that $f$ is a $G$-equivariant map; it is also required that
$G$ acts freely on $Y \setminus Z$.
When restricted to $Y \setminus Z$, the map $f$ becomes a $G$-equivariant map
from $X$ to $Y \setminus Z$. Finally, one shows that no such maps can
exist---contradicting the assumption that our desired object does not exist.
This last part of proving the non-existence of equivariant maps
is where tools from Algebraic Topology typically come into play.

\subsubsection{Blagojevi\'c's and Dimitrijevi\'c Blagojevi\'c's Theorem under the CS-TM Scheme}\label{sec:k_fan}

The theorem of Blagojevi\'c and Dimitrijevi\'c Blagojevi\'c
is an equipartition result of $3$-fans on the sphere; 
we now illustrate how their proof fits into the CS-TM scheme.
For the following definitions regarding $k$-fans we follow the exposition
in B\'ar\'any's and Matou{\v{s}}ek's paper~\cite{k-fan}.
The study of equipartition results using $k$-fans was initiated
by Akiyama Kaneko, Kano, Nakamura, Rivera-Campo, Tokunaga and Urrutia
in \cite{jorge}.

A \emph{$k$-fan} in the plane is a set of $k$ infinite rays,
that emanate from the same point. This point is called 
its \emph{apex}. A $k$-fan in the plane can also be a set of $k$ parallel
lines. Given a $k$-fan $\gamma$ in the plane, we call the connected
open regions of $\mathbb{R}^2 \setminus \gamma$ 
the \emph{wedges} of $\gamma$. In the case where 
$\gamma$ consists of parallel lines, a wedge is also the union of the two open regions of  
$\mathbb{R}^d \setminus \{\gamma\}$ that are bounded by a single
line.

The inclusion of $k$ parallel lines in the definition
of $k$-fans and the last exception in the definition of its wedges
may seem awkward. However, $k$-fans consisting of rays
and $k$-fans consisting of parallel lines are closely related.
This connection will become clear once we consider
$k$-fans in the sphere and their connection with $k$-fans in the plane.

A \emph{$k$-fan} in the two dimensional \emph{sphere} $S^2$ is a set
of $k$ great semicircles that emanate from the same two antipodal points. (Recall that $S^2$
is a \emph{two}-dimensional surface, but it is normally regarded as embedded in $\mathbb{R}^3$.)
These two points are called its apices.  Given a $k$-fan $\gamma$ in $S^2$, we call the connected
open regions of $S^2 \setminus \gamma$ the \emph{wedges} of $\gamma$. See Figure~\ref{fig:wedge}.

\begin{figure}
  \begin{center}
   \includegraphics[width=0.5\textwidth]{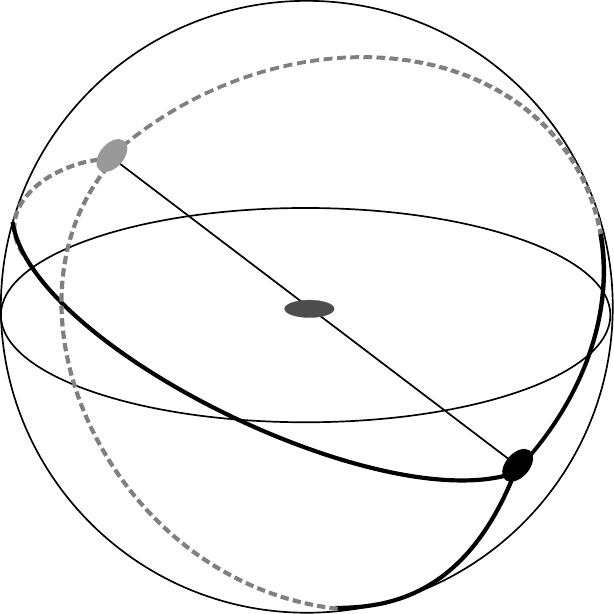}
\end{center}
\caption{A $3$-fan on $S^2$.}
\label{fig:wedge}
\end{figure}
The connection between $k$-fans in the sphere
and $k$-fans in the plane is given by the following map from
the open southern hemisphere of $S^2$ to $\mathbb{R}^2$.
Assume that $S^2$ lies on $\mathbb{R}^3$ and identify
$\mathbb{R}^2$ with a horizontal plane lying below $S^2$.
From the center of $S^2$ project every point on the southern
hemisphere of $S^2$ to $\mathbb{R}^2$. Let $\pi$ be this
map. The image under $\pi$ of a $k$-fan for which all its great semicircles
intersect the open southern hemisphere of $S^2$ is a $k$-fan
in $\mathbb{R}^2$. If the apices of this $k$-fan are on the
equator of $S^2$ then the image of this $k$-fan corresponds
to a set of $k$ parallel lines in $\mathbb{R}^2$. Conversely
the preimage under $\pi$ of a $k$-fan in $\mathbb{R}^2$ corresponds
to a $k$-fan in $S^2$ for which all its semicircles 
intersect the open southern hemisphere. 
This connection between $k$-fans in the sphere and in the plane
allows us to translate partition results by $k$-fans
in the sphere to partition results by $k$-fans in the plane. 

Let $R$ and $B$ be two finite Borel measures on $S^2$ and
$\alpha, \beta >0$ be two real numbers such that $2\alpha+\beta=1$. 
The result of \cite{equivariant} states that there exists
a $3$-fan such that two of its wedges have an $\alpha$ proportion
of $R$ and $B$, and the remaining wedge contains a $\beta$
proportion of $R$ and $B$.  Keeping up with our convention we reformulate 
that statement in terms of areas of open sets. 

\begin{theorem}\label{thm:fan}\textbf{(Theorem 3.2 in \cite{equivariant})}
Let $R$ and $B$ be two open sets on $S^2$, and let
$\alpha, \beta >0$ be two real numbers such that $2\alpha+\beta=1$. 
Then there exists a $3$-fan on $S^2$ such that its corresponding
wedges $W_1, W_2$ and $W_3$ satisfy:
\begin{itemize}

\item $\vol(W_1 \cap R)/\vol(R)=\vol(W_1 \cap B)/\vol(B)=\alpha$,
\item $\vol(W_2 \cap R)/\vol(R)=\vol(W_2 \cap B)/\vol(B)=\beta$ and
\item $\vol(W_3 \cap R)/\vol(R)=\vol(W_3 \cap B)/\vol(B)=\alpha$.
\end{itemize}
\end{theorem}

The proof of Theorem~\ref{thm:fan} in~\cite{equivariant} is done 
when $\alpha$ and $\beta$ are rational numbers. Afterwards, using standard
arguments it is shown that Theorem~\ref{thm:fan} holds when $\alpha$ and $\beta$ 
are real numbers. So assume that $\alpha$ and $\beta$ are rational numbers.
Let $k, a_1, a_2$ be natural numbers
such that $\alpha=\frac{a_1}{k}$ and $\beta=\frac{a_2}{k}$. Note that
$2a_1+a_2=k$.

\subsubsection*{Configuration Space}

We consider $k$-fans rather than $3$-fans. 
We need to ``orient'' the set of $k$-fans, similar to as 
we did with the halving planes in Section~\ref{sec:ham}.
Let $\gamma$ be a $k$-fan in $S^2$.  Recall that $\gamma$
has two apices, and $k$ great semicircles. The orientation is given 
by choosing a tuple $(x,C)$ where $x$ is one of its apices, and $C$ is 
one of its great semicircles.
 Once an orientation $(x,C)$ is chosen for $\gamma$,
the $k$ great semicircles $C_1,\dots, C_k$ of $\gamma$ are assumed to 
be sorted counterclockwise around $x$, with $C_1=C$ being the first
one. For $1 \le i \le k-1$, let $\gamma_i$ be the wedge of $\gamma$ 
bounded by $C_i$ and $C_{i+1}$, and let $\gamma_k$ be the wedge bounded by
$C_k$ and $C_1$.

Let $M$ be the set of all $k$-fans that equipartition $R$. 
That is, all $k$-fans $\gamma$ such that
$\vol(\gamma_i \cap R)=\vol(R)/k$ for $i=1,\dots, k$.
Let $\gamma \in M$ be such a $k$-fan and let $\ell$ be the straight line
through its two apices. Note that each semicircle
$C_i$ of $\gamma$ can be moved around a maximal interval $I_i$ (of possibly one point) of great semicircles around $\ell$,
so that $\vol(\gamma_i \cap R)$ does not change. If $I_i$ has more than one point, define $\varphi_i$ as the wedge bounded
by the great semicircles at
the endpoints of $I_i$. In particular, if $I_i$ has more than one point then $\vol(\varphi_i \cap R)$ is equal to zero.
Let $\gamma^*$ be the $k$-fan that is obtained from $\gamma$ by placing each $C_i$ at the midpoint
of $I_i$. Note that $\gamma^* \in M$. 

The configuration space is the subset of $M$, given by $M^*:=\{\gamma^*:\gamma \in M\}$.
We parametrize this space by assigning a pair of orthogonal unit vectors to each $\gamma \in M^*$ as follows. 
Assume that $S^2$ is of radius one.  Assign to $\gamma$ the tuple $(\vec{x}, \vec{y})$
where $\vec{x}$ is the unit vector with endpoint at $x$, and $\vec{y}$ is the
unit vector orthogonal to $\vec{x}$ whose endpoint lies in $C_1$.
Note that since $\gamma$ equipartitions $R$, and $\gamma$ is in $M^*$, $x$ and $C_1$ determine
all of $\gamma$. Thus, a tuple $(\vec{x}, \vec{y})$ of orthogonal 
vectors corresponds to
exactly one of these $k$-fans. 
With this correspondence in mind, the 
configuration space is the space, $V_2(\mathbb{R}^3)$,
 of  all tuples of unit orthogonal
vectors in $\mathbb{R}^d$. (In the literature $V_2(\mathbb{R}^3)$
is known as the Stiefel manifold~\cite{stiefel}.) 

We specify a group action on $V_2(\mathbb{R}^3)$. Let $\phi$ be the homeomorphism
of $V_2(\mathbb{R}^3)$ that sends $\gamma$ to $(-x,C)$, and let $\psi$
be the homeomorphism that sends $\gamma$ to $(x,C_2)$. The
group generated by $\phi$ and $\psi$ is isomorphic 
to the dihedral group $\mathbb{D}_{2k}$---the
group of symmetries of the regular $k$-gon. This is clear
once one realizes that $\phi$ corresponds to a reflection of the regular 
$k$-gon, and that $\psi$ corresponds to a clockwise rotation 
of the $k$-gon by one. We assume that $\mathbb{D}_{2k}$ acts
on $V_2(\mathbb{R}^3)$ via $\psi$ and $\phi$; note that this
action is free.

\subsubsection*{Test Space and Test Map}

Let $Y$ be the linear subspace of $\mathbb{R}^k$ defined
by 

\[Y:=\{(x_1,\dots,x_k) \in \mathbb{R}^k:x_1+\dots+x_k=0\}.\]
Let $Z$ be the subspace of $Y$ defined by the equations

\begin{IEEEeqnarray*}{lCll}
  x_1 & +\dots+ & x_{a_1} & = 0,\\
  x_{a_1+1} & +\dots+& x_{a_1+a_2} &  = 0, \\
  x_{a_1+a_2+1} & +\dots+& x_{k} &  = 0.
\end{IEEEeqnarray*}

Let $f:V_2(\mathbb{R}^3) \to Y$ be the map defined by

\[ f(\gamma)=(\vol(\gamma_1\cap B)-\vol(B)/k,\dots, \vol(\gamma_k \cap B)-\vol(B)/k).\]

Let $\gamma'$ be the $3$-fan with apex $x$ and with great semicircles $C_1, C_{a_1}$ and $C_{a_1+a_2}$.
Note that if $f(\gamma) \in Z$ then:

\begin{IEEEeqnarray*}{cccccc}
 \vol(\gamma_1' \cap B)/ \vol(B) & = & \sum_{i=1}^{a_1} \vol(\gamma_i \cap B)/\vol(B)~       & = & ~\alpha, &\\ 
 \vol(\gamma_2' \cap B)/ \vol(B) & = & \sum_{i=a_1+1}^{a_1+a_2} \vol(\gamma_i \cap B)/\vol(B)~ & = & ~\beta, & \textrm{ and} \\
 \vol(\gamma_3' \cap B)/  \vol(B)& = & \sum_{i=a_1+a_2}^{k} \vol(\gamma_i \cap B)/ \vol(B)~   & = & ~\alpha.& 
\end{IEEEeqnarray*}

In this case, $\gamma'$ is the desired $3$-fan of Theorem~\ref{thm:fan}. 
We now equip $Y$ with a free $\mathbb{D}_{2k}$-action. Let $\phi'$ be the homeomorphism
of $Y$ that sends $(x_1, x_2,\dots, x_{k-1}, x_k)$ to $(x_k, x_{k-1},\dots, x_2, x_1)$, and
let $\psi'$ be the homeomorphism of $Y$ that sends $(x_1,x_2,\dots, x_k)$ to 
$(x_2,\dots,x_k, x_1)$. The group generated by $\phi'$ and $\psi'$ is isomorphic to 
$\mathbb{D}_{2k}$. We assume that $\mathbb{D}_{2k}$ acts on $Y$ via $\phi'$ and
$\psi'$. Note that this action is free and $f$ is a $\mathbb{D}_{2k}$-equivariant
map from $V_2(\mathbb{R}^3)$ to $Y$. If no $k$-fan is mapped by $f$ to $Z$ then $f$ 
is an equivariant map from $V_2(\mathbb{R}^3)$ to $Y \setminus Z$. In \cite{equivariant},
Blagojevi\'c and Dimitrijevi\'c Blagojevi\'c prove that no such
map exists. The proof of this last part is far from trivial. Indeed, it is
the gist of the proof of Theorem~\ref{thm:fan}--we have merely presented the prelude.
Unfortunately, a detailed account of this is
 beyond an expository account.

\section{Proof of the Balanced Island Theorem}\label{sec:hobby}
We are now ready to prove Theorem~\ref{thm:hobby}. We show the first
case in Lemma~\ref{lem:gen} and the second case in Lemma~\ref{lem:n+1}.

A way to prove partition theorems on points sets
from similar partition theorems on finite Borel measures (or areas of open sets in our case)
is the following.
First enlarge each point to a disk of radius $\varepsilon>0$, and use the measure
theorem to find a solution. Then let $\varepsilon$ tend to zero and show
that the limit of these solutions exists and that it is  the desired
solution for point sets. We applied this approach 
in Section~\ref{sec:ham}, when we sketched how to obtain the
 discrete version (Theorem~\ref{thm:ham_d}) of the Ham Sandwich theorem
from its continuous version (Theorem~\ref{thm:ham_c}). We follow
this approach again in the proof of Lemma~\ref{lem:gen}.

\begin{lemma}\label{lem:gen} 

Let $S$ be a set of $r$ red points and $b$ blue points
in the plane. Then for every $\alpha \in \left [ 0,\frac{1}{2} \right ]$
there exists a convex set containing exactly $\lceil \alpha r\rceil$ red points
and exactly $\lceil \alpha b \rceil$ blue points of $S$. Moreover, this convex
set is either a convex wedge or a strip.
\end{lemma}
\begin{proof}

Assume without loss of generality that $\alpha$ and $\beta$ are rationals.
We project the open southern hemisphere of $S^2$ to $\mathbb{R}^2$ 
as in Section~\ref{sec:k_fan}. Assume that $S^2$ lies on $\mathbb{R}^3$ and identify
$\mathbb{R}^2$ with an horizontal plane lying below $S^2$. 
From the center of $S^2$ project every point on the southern
hemisphere of $S^2$ to $\mathbb{R}^2$. Let $\pi$ be this map.

Let $\varepsilon > 0$.
On every red point 
place an open disk of radius $\varepsilon$ centered at this point; 
let $R_\varepsilon$ be the union of all these disks. Likewise, 
on every blue point  place an open disk of radius $\varepsilon$ centered at this point; 
let $B_\varepsilon$ be the union of all these disks. 
Let $\gamma_\varepsilon$ be the $3$-fan of $S^2$ given by 
Theorem~\ref{thm:fan} for $R:=\pi^{-1}(R_\varepsilon)$, $B:=\pi^{-1} (B_\varepsilon)$, $\alpha$ and 
$\beta:=1-2\alpha$. Choose $\varepsilon$ small enough so that  $R$ and $B$ lie on the southern hemisphere of $S^2$;
if necessary, perturb $\gamma_\varepsilon$ so that the three semicircles
of $\gamma_\varepsilon$ intersect the southern hemisphere of $S^2$. Note that $\pi(\gamma_\varepsilon)$
is a $3$-fan in $\mathbb{R}^2$. 

The $3$-fan $\gamma_\varepsilon$ defines three wedges: two of which contain
an $\alpha$ proportion of the area of $R$ and $B$; the remaining wedge, $W_\beta$,
contains a $\beta$ proportion of the area $R$ and $B$. Assume that 
$\gamma_\varepsilon$ is oriented so that $W_\beta$ is bounded
 by the first and the second semicircle. At least one of the two wedges that contain
an $\alpha$ proportion of the area of $R$ and $B$ is convex when projected
under~$\pi$. Of these two wedges, let $W_\varepsilon$ be the first wedge clockwise from $W_\beta$ that is convex
when projected under~$\pi$.

Note that once the apex, $p$, and the first semicircle, $C_1$, 
of $\gamma_\varepsilon$ are fixed, all of $\gamma_\varepsilon$ is determined.
This implies that once $\gamma_\varepsilon$ is oriented we can identify 
it with a tuple of unit orthogonal vectors. 
Assume that $S^2$ is the 
unit sphere centered at the origin. Set the first vector
of the tuple to be $p$ and the second to be the unit vector orthogonal 
to $p$ that lies on $C_1$. It can be verified that this space, $V_2(\mathbb{R}^3)$, 
of tuples of orthogonal unit vectors in $\mathbb{R}^3$ is compact.  

Consider any sequence of $\varepsilon$'s that converges to zero; the corresponding
sequence of $\gamma_\varepsilon$'s has a limit point  $\gamma$.
Let  $\{\gamma_{\varepsilon_i}\}_{i=1}^\infty$ be a subsequence of this sequence that converges
to $\gamma$ and let $W$ be the wedge
of $\gamma$ that $\{W_{\varepsilon_i}\}_{i=1}^\infty$ converges to. $W$~is convex  since
each of the terms in $\{\gamma_{\varepsilon_i}\}_{i=1}^\infty$ is convex when projected
under~$\pi$. 

Note that since  $S$ is in general position no three points $\pi^{-1}(S)$ are in a
common great semicircle of $S^2$. Furthermore, each of the terms in 
$\{W_{\varepsilon_i}\}_{i=1}^\infty$ contain an $\alpha$ proportion of the
area of $R$ and $B$. This implies that the closure of $W$ (in $S^2$)
contains at least  $\lceil \alpha r \rceil$ and at most
 $\lceil \alpha r \rceil+1$ red points of $\pi^{-1}(S)$. By
the same token, it contains at least  $\lceil \alpha b \rceil$ and at most
 $\lceil \alpha b \rceil+1$ blue points of $\pi^{-1}(S)$.
If $W$ has one more red point than  $\lceil \alpha r \rceil$,
then a red point lies in one of the bounding semicircles of $W$. Similarly, 
if $W$ has one more blue point than  $\lceil \alpha b \rceil$,
then a blue point lies in one of the bounding semicircles of $W$.
In both cases a small perturbation of $W$ ensures that it
contains exactly $\lceil \alpha r \rceil$ red points
and exactly $\lceil \alpha b \rceil$ blue points.
The convex wedge $\pi(W)$ is the desired convex set
and the result follows.
\end{proof}

\begin{figure}
  	\begin{center}
   	\includegraphics[scale = 0.65, page = 2]{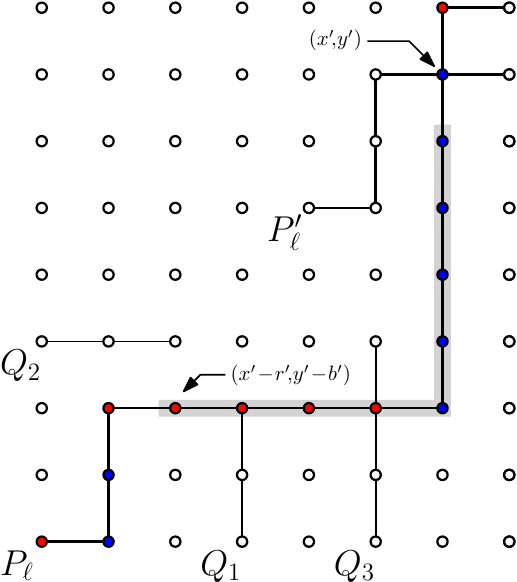}
	\end{center}
\caption{Illustration of the proof of Lemma~\ref{lem:n+1} for $r=7$ and $b=8$. 
	The points of $P_\ell$ are colored like the according points of $S$. 
	The part of $P_\ell$ that corresponds to the desired subset of $S$ is marked with a gray background.}
\label{fig:grid}
\end{figure}

\begin{lemma}\label{lem:n+1}
Let $S$ be a set of $r$ red points and $b$ blue points in the plane. Then 
there exists a strip containing exactly $\left \lceil \frac{r+1}{2}\right \rceil$ red points and exactly $\left \lceil \frac{b+1}{2}\right \rceil$ blue points of $S$.
\end{lemma}

\begin{proof}
Let $r':=\left \lceil \frac{r+1}{2} \right \rceil$ and $b':=\left \lceil \frac{b+1}{2}\right \rceil$.
%%% P_\ell was used doubly (for set on \ell and path) => set on \ell := S_\ell
For a given oriented line $\ell$, let $S_\ell$ be the orthogonal projection of $S$ to $\ell$.
Assume that $\ell$ is such that no two points of $S$ are projected to the same point in $S_\ell$. 
Further assume that the points of $S_\ell$ are sorted by the order in which they appear on $\ell$. 
Note that if $P_\ell$ contains an interval (i.e., a contiguous subsequence of points) of exactly $r'$ red points and exactly $b'$ blue points,
then there exists a strip bounded by two lines orthogonal to $\ell$ and containing exactly $r'$ red points and $b'$ blue points of $S$.

Let $G$ be the $(r+1) \times (b+1)$ integer grid graph with vertex set $\{(i,j): 0 \le i \le r \textrm{ and } 0 \le j \le b\}$, 
	in which two vertices are adjacent if in one of their coordinates they are equal and in the other they differ by one. 
We assign a path $P_\ell:=((i_1,j_1),\dots,(i_{r+b},j_{r+b}))$ in $G$ to $\ell$.
We define it as follows. 
The first vertex $(i_1,j_1)$ is equal to $(0,0)$. 
For $k\ge 2$, the $k$-th vertex $(i_k,j_k)$ is equal to $(i_{k-1}+1,j_{k-1})$ if the $(k-1)$-st element of $S_\ell$ is red, and equal to  $(i_{k-1},j_{k-1}+1)$ if it is blue.
Note that $P_\ell$ always ends at $(r,b)$.

Let $P_\ell'$ be the translation of $P_\ell$ by the vector $\left ( r', b' \right )$.
That is, $P_\ell'$ has length $r+b-r' - b'$, and the $k$-th vertex of $P_\ell'$ has coordinates $\left (i_k+r' ,j_k+b' \right )$. 
%%% (x,y) was used doubly (for P_\ell \cap P_\ell' and for P_\ell  \cap Q_1) => P_\ell \cap P_\ell':= (x',y')
If $P_\ell$ and $P_\ell'$ have a common vertex $(x',y')$ then $P_\ell$ contains the vertex $\left ( x'-r', y'-b' \right )$. 
%This implies that between the $\left (x'-r'+y'-b' \right )$-th and the $(x'+y')$-th point of $P_\ell$ there are exactly $r'$ red points and exactly $b'$ blue points.
This implies that in the contiguous sequence from the $\left (x'-r'+y'-b' \right )$-th to the $(x'+y'-1)$-st point of $S_\ell$ there are exactly $r'$ red and $b'$ blue points.
Hence it suffices to show that $P_\ell$ and $P_\ell'$ intersect.

We may assume that $P_\ell$ does not contain the vertex $\left ( r-r',b-b' \right )$, since it would imply that $P_\ell'$ also ends at $(r,b)$. 
Therefore, $P_\ell$ must intersect the path 
	$Q_1:=\left ( r-r',0 \right )$, $\left( r-r',1 \right )$, $\dots$, $\left ( r-r', b-b' -1 \right )$ or the path 
	$Q_2:=\left ( 0,b-b' \right )$, $\left( 1, b-b' \right )$, $\dots$ , $\left ( r-r'-1, b-b' \right )$.
Without loss of generality, assume that $P_\ell$ intersects $Q_1$ (if not we interchange the colors), and let $(x,y)$ be an intersection point of $P_\ell$ and $Q_1$.

Suppose first that $P_\ell$ intersects the path $Q_3:=$ $\left ( r-r'+2,0 \right )$, $\left( r-r'+2,1 \right )$, $\dots$ , $\left ( r-r'+2, b-b' \right )$. 
Then, since $P_\ell$ ends at $(r,b)$ it must intersect the subpath of $P_\ell'$ that starts at $(r',b')$ and ends at $(r'+x,b'+y)$, see Figure~\ref{fig:grid}. 

Assume now that $P_\ell$ does not intersect $Q_3$.
Therefore, $P_\ell$ contains the vertices $(r-r'+1,b-b'-1)$, $(r-r'+1,b-b')$ and  $(r-r'+1,b-b'+1)$. 
In particular, the $(r+b-r'-b')$-th and the $(r+b-r'-b'+1)$-th point of $S_\ell$ are blue. 
Move $\ell$ continuously until it reaches a line parallel to it but with opposite orientation. 
Throughout the motion, contiguous elements in $S_\ell$ exchange their positions until $S_\ell$ is inverted. 
During these exchanges, $P_\ell$ also changes somewhat continuously: at each step, one rightwards-upwards corner may become upwards-rightwards, or vice versa. 
This implies that at some line $\ell^*$ the $(r+b-r'-b')$-th or the $(r+b-r'-b'+1)$-th point of $P_{\ell^*}$ is red and $P_{\ell^*}$ and $P_{\ell^*}'$ intersect.
%the result follows.
\end{proof}

\section{Algorithms} \label{sec:algo}

A drawback of using topological methods is that they tend  to provide
existential rather than constructive proofs. This is the case
for the proof of Theorem~\ref{thm:hobby}. In this section we give polynomial time
algorithms to find balanced islands.

Our strategy is to use the fact that the convex sets in 
Theorem~\ref{thm:hobby} are either strips or wedges (Lemmas~\ref{lem:gen} and \ref{lem:n+1}). We discretize
the space of candidates (wedges or strips) and efficiently visit them.
 If we are looking for a wedge we first discretize
the set of possible apices; if we are looking for a strip
we first discretize the set of possible directions for the boundaries of the strip.
These two steps are very similar. Indeed,
strips and wedges are the same object on the sphere; 
fixing a direction of a strip is equivalent to fixing
an apex for the corresponding wedge on the sphere (see Section~\ref{sec:k_fan}). 

Suppose that we are looking for a wedge.  We know that
the solution must have exactly $k:=\lceil \alpha r\rceil+\lceil \beta b\rceil$ points of $S$ (regardless of color).
Suppose that a finite set of candidates apices has been chosen, and
that $p$ is the first candidate apex visited.
The wedges with apex $p$ containing exactly $k$ points
of $S$ can be discretized as follows. Sort $S$ clockwise by angle
around $p$  in $O(n \log n)$ time. Consider the set of 
all intervals  of $k$ contiguous points of $S$ in this ordering; 
the set of candidate wedges with apex $p$ are those wedges 
whose bounding rays pass through the endpoints
of these intervals. Note that for every wedge with
apex $p$ and containing $k$ points of $S$, there exists
a candidate wedge containing exactly the same points of $S$. 
This set of wedges around $p$ can be constructed in  $O(n \log n)$ 
time, and as we construct them we can check whether any of them
is balanced and convex. The candidate apices will then be visited in such a way so as
to minimize the changes in this set of candidate wedges; the only possible
changes are that two points transpose in the order around
the apex, or that a wedge ceases to be or becomes convex. In particular, 
at most two wedges can change: the two wedges whose  associated intervals
have endpoints at the points that transpose or the wedge that ceases to be or becomes 
convex. As a result, we can update our set of candidates in constant time when we go 
from one candidate apex to the next.  

When we are looking for strips we proceed in a similar way.
First we discretize the the set of possible directions for the strips.
We visit the first direction $\ell$ and sort the points in the orthogonal
direction to $\ell$ in $O(n \log n)$ time. Like before,
the candidate strips are the set of intervals containing $k$ points in this
ordering. When we visit the next direction the only change that occurs is that 
two points transpose in this ordering. Again, we can update our set
of candidate strips in constant time. In the proof of 
Lemmas~\ref{lem:wedge_alg} and \ref{lem:strips_alg}, we detail how to compute
and visit these sets of candidates.

\begin{lemma}\label{lem:wedge_alg}
Let $S$ be a set of $n$ points in general position in the plane, $r$ of which are red and $b$ of which are blue.
Then for any $r'\le r$ and $b'\le b$, finding a convex wedge 
containing exactly $r'$ red points and exactly $b'$ 
points of $S$, or determining that no such wedge exists,
can be done in $O(n^4)$ time.
\end{lemma}
\begin{proof}

Let $\mathcal{A}$ be the line arrangement generated by the set of lines 
that pass through every pair of points in $S$. As there are $O(n^2)$ such lines, 
$\mathcal{A}$ can be constructed in $O(n^4)$ time using standard
algorithms for constructing line arrangements. Also note that $\mathcal{A}$
has $O(n^4)$ cells.

Let $C$ be a cell of $\mathcal{A}$ and let $p$ be 
a point in its interior. Starting at
an arbitrary point of $C$, sort the points of $S$ clockwise by angle around $p$. 
 This is done in $O(n \log n)$ time;
let $S_p$ be this sorted set.

Let $k:=r'+b'$.
For $1\le i \le n$, let $I_i$ be the interval of $S_p$
that starts at the $i$-th point and ends at the $(i+k-1)$-th
point of $S_p$ (modulo $n$). We compute the number of red and blue points in each $I_i$
in $O(n)$ time, by computing this parameter for $I_1$ and updating it
when we move from $I_i$ to $I_{i+1}$.  We keep pointers from the first and last vertices
of $I_i$ to itself. Let $W_i$ be the wedge with apex $p$, that contains $I_i$, and
whose bounding rays pass through the first and last vertex of $I_i$. We also keep
a record of whether $W_i$ is convex. Note that neither the number of red and blue points of $S$
 inside $W_i$, nor whether it is convex, depend on the choice of $p$ within $C$. 

We visit all the cells in $\mathcal{A}$ by doing a DFS search in
$O(n^4)$ time on its
dual graph, so that we move between adjacent cells. We choose
a point $p$ in the interior of each of them. 
The only changes than can occur is that two consecutive
points of $S_p$ transpose or that a wedge ceases to be or becomes convex.
This is determined by which line of $\mathcal{A}$ was crossed
when visiting the next cell. By knowing which pair of points define
this line,  we can update the corresponding $W_i$'s in constant time.

Consider any convex wedge $\gamma$ with apex $q$ that contains $k$ points
of $S$. Some cell of $\mathcal{A}$ contains
$q$ and at some point in our algorithm we chose a point $p$ in this cell. 
One of the candidate wedges  with apex $p$ contains exactly the same 
points as $\gamma$ and the result follows.
\end{proof}

\begin{lemma}\label{lem:strips_alg}
Let $S$ be a set of $n$ points in general position in the plane,
$r$ of which are red and $b$ of which are blue.
Then for any $r'\le r$ and $b'\le b$, finding a strip 
containing containing exactly $r'$ red points and exactly $b'$ blue
points of $S$, or determining that no such strip exists,
can be done in $O(n^2 \log n)$ time.
\end{lemma}
\begin{proof}

Let $\mathcal{L}$ be the set of lines generated by every pair of 
points in $S$. Sort the lines in $\mathcal{L}$ by slope in $O(n^2\log n)$
time. Let $\ell_1$ be the first line of $\mathcal{L}$. Project $S$ orthogonally
to $\ell_1$ and let $S_{\ell_1}$ be this set. Note that strictly speaking,
 the pair of points that define $\ell_1$ are mapped
to the same point. We wish to avoid this, so we actually choose a line with a slightly larger slope
than $\ell_1$ (but smaller than the next line in $\mathcal{L}$).
 We will make this choice each time we visit a line in $\mathcal{L}$.

 Sort in $O(n \log n)$ time 
the points in $S_{\ell_1}$ by the order in which they appear on $\ell_1$. 
Note that if $S_{\ell_1}$ contains an interval of  exactly $r'$ red points
and exactly $b'$ blue points
then there  exists a strip bounded by two lines, orthogonal to $\ell_1$, containing 
exactly $r'$ red points and 
exactly $b'$ blue points of $S$.

Let $k:=r'+b'$. For $1\le i \le n-k+1$, let $I_i$ be the interval of $S_{\ell_1}$
that starts at the $i$-th point and ends at the $(i+k-1)$-th
point of $S_{\ell_1}$. We compute the number of red and blue points in each $I_i$
in $O(n)$ time, by computing this parameter for $I_1$ in $O(n)$ time
and updating it when we move from $I_i$ to $I_{i+1}$.
We keep pointers from the first and last vertices
of $I_i$ to itself. 

We visit the $\ell_j$'s in order, while maintaining 
the $I_i$'s and their respective number of red and blue points. This can be
done in constant time per line. At each step only two consecutive
points of $S_{\ell_j}$ interchange their positions. The only intervals
that change their endpoints---and thus their number of red and blue
points---are precisely the two intervals that start and end at these
points.

Note that every strip can be rotated without changing
the points of $S$ it contains until its bounding lines
are orthogonal to a line in $\mathcal{L}$. Therefore,
finding a strip containing
exactly $r'$ red points and exactly $b'$ blue points of $S$, or
determining that no such a strip exists,
can be done in $O(n^2 \log n)$ time. 
\end{proof}

We now state the algorithmic version of Theorem~\ref{thm:hobby}.
\begin{theorem}\label{thm:hobby_alg}
Let $S$ be a set of $n$ points in general position in the plane,
$r$ of which are red and $b$ of which are blue; 
let $\alpha \in \left [ 0,\frac{1}{2} \right ]$, then:
\begin{enumerate}
  \item  a convex set containing exactly $\lceil \alpha r\rceil$ red points
and exactly $\lceil \alpha b \rceil$ blue points of $S$ can be found in $O(n^4)$ time;

  \item a convex set containing exactly $\left \lceil \frac{r+1}{2}\right \rceil$ red points
and exactly $\left \lceil \frac{b+1}{2}\right \rceil$ blue points of $S$ can be found in $O(n^2\log n)$ time.
\end{enumerate}
\end{theorem}
\begin{proof}
The existence of these convex sets follows from Lemmas~\ref{lem:gen} and \ref{lem:n+1}.
The running times of the algorithms to find them follow from  Lemmas~\ref{lem:wedge_alg} 
and \ref{lem:strips_alg}.
\end{proof}

\subsection{Balanced Islands in $O(n \log n)$ Time}

The running times of the previous algorithms can be improved significantly for many values of $\alpha$.  
For example, Lo and Steiger~\cite{ham_alg} gave an optimal $O(n)$ time algorithm for finding a Ham Sandwich cut for $S$, by this giving an optimal algorithm for $\alpha=\frac{1}{2}$.
Using the results from the following lemmas, we will present a significantly improved algorithm for a large range of values $\alpha$ in Theorem~\ref{thm:alg_2}.  
%many other instances of $\alpha$.
%In Theorem~\ref{thm:alg_2} we show many other instances where the running times of our algorithms can also be significantly improved.

To this end, we first introduce the concept of weighted islands. 
Assume that every red point in $S$ is given a positive weight of $1/r$ and every blue point in $S$ is given a negative weight of $-1/b$. 
For a given island of $S$ let its weight be the sum of the weight of its points.
To obtain a balanced island for a given $\alpha \in [0,\frac{1}{2}]$, we want 
	an island of weight $\lceil \alpha r\rceil/r - \lceil \alpha b\rceil/b$ 
	and containing $\lceil \alpha r\rceil + \lceil \alpha b\rceil$ points. 
If $\alpha r$ and  $\alpha b$ are both integers then the weight of this island is equal to zero. 
Otherwise, it is as close to zero as possible, among the islands of $S$ with  $\lceil \alpha r\rceil + \lceil \alpha b \rceil$ points.
If an island has weight larger than $\lceil \alpha r\rceil/r - \lceil \alpha b\rceil/b$, we call it \emph{positive}; 
if it has weight smaller than this value we call it \emph{negative}. 

Suppose that we have found a positive island $I$ and a negative island $J$, both with $\lceil \alpha r\rceil + \lceil \alpha b \rceil$  points of $S$. 
A promising approach would be to move ``continuously'' from $I$ to $J$, so that somewhere in the middle 
	we find a balanced island of $\lceil \alpha r\rceil + \lceil \alpha b \rceil$ points.
This indeed can be done. In \cite{island_graph} Bautista-Santiago et al.\ defined a graph whose vertices are all the islands of $S$ of a given size~$k$, 
	where two of them are adjacent if their symmetric difference has a fixed cardinality $\ell$. 
They showed that under mild assumptions on $k$ and $\ell$ that this graph is connected. 
In our case we have $k=\lceil \alpha r\rceil + \lceil \alpha b \rceil$ and $\ell=2$. 
A path from $I$ to $J$ in the resulting graph is our desired sequence; somewhere in the middle of such a sequence there is a balanced island. 
We show how to compute this sequence.

\begin{lemma}\label{lem:island_graph}
Let $S$ be a set of $n$  points in the plane and 
let $I$ and $J$ be islands of $S$ of $k$ points each.
Then there exists a sequence of $O(n)$ islands, starting at $I$ and ending at $J$, such that the symmetric difference between two consecutive islands is a pair of points. 
This sequence can be computed in $O(n \log n)$ time.
\end{lemma}

\begin{proof}
Let $G$ be the graph whose vertices are all the islands of $S$ with $k$ points, two of which are adjacent in $G$ if their symmetric difference is a pair of points. 
We look for a path of linear length from $I$ to $J$ in $G$. 
Without loss of generality, assume that no two points of $S$ have the same $x$-coordinate.
We sort the points in $S$ by their $x$-coordinate. 
The subsets of $S$ consisting of $k$ consecutive points in this ordering are islands of $S$ (thus vertices of $G$).
Let $G'$ be the subgraph of $G$ induced by these islands.
Note that $G'$ is a path of length at most $n$, and can be computed in $O(n\log n)$ time.  
To complete the proof, we show how to compute a path of linear length from $I$ to a vertex of $G'$ in $O(n \log n)$ time. 
%We only show how to compute this path from $I$, since 
A path from $J$ to a vertex of $G'$ can be computed in a similar way. 

Compute the convex hull, $C$, of $I$ in $O( n\log n)$ time.  
Let $S'$ be the points of $S \setminus I$ that lie in the vertical strip between the leftmost and rightmost point of $I$. 
Initialize a priority min-queue $Q$ with the vertices of $S'$.  
Store the points in $Q$ according to their shortest distance to $C$. 
This distance can be computed in $O(\log n)$ time per point, by doing a binary search on  $C$.  
Set $I_1:=I$. 

Assume that $I_i$ has  been computed and that $p_i$ is its rightmost point.  
We extract points from $Q$ until we find a point $q$ that is to the left of $p_i$. 
Note that $q$ is the point to the left of $p_i$ closest to $C$. Set $I_{i+1}:=I_i \cup \{q\} \setminus{p_i}$. 
The rightmost point $p_{i+1}$ of $I_{i+1}$ can be computed in constant time, since it is either $q$ or the first point of $I$ to the left of $p_i$.  
If $Q$ is empty or no such point is found then $I_i$ is an island of $G'$ and we are done.  
Note that by construction the symmetric difference between $I_i$ and $I_{i+1}$ is a pair of points. 
It remains to show that the $I_i$'s are islands of $S$.

Suppose that some $I_i$ is not an island of $S$. 
Then there exists a point $p \in S \setminus I_i$ contained in the convex hull of $I_i$.  
By Caratheodory's theorem there exist three points $q_1, q_2$ and $q_3$ of $I_i$ that contain $p$ in their convex hull. 
One of these points, say $q_1$, is farther away from $C$ than $p$.  
In particular $q_1$ is not in $I$. Therefore, $q_1$ was added to create some $I_j$ with $j<i$.  
When $I_j$  was created, $q_1$ was chosen because it was the point of $S'\setminus I$ closest to $C$ and to the left of $p_{j-1}$. 
This is a contradiction since $p$ is also a point of $S'\setminus I$ to the left of $p_{j-1}$, and it is closer to $C$ than $q_1$.
\end{proof}

The algorithm to find a balanced wedge in Theorem~\ref{thm:hobby_alg} computes a set of candidate apices, such that
one of them is guaranteed to be the apex of a balanced wedge. 
The set of candidates however is quite large; it has size $\Theta(n^4)$. 
If $\lceil \alpha r\rceil + \lceil \alpha b \rceil$ is not too large, we can somehow reduce this set of candidates 
to a single point $p$. This point has the property that either there is a balanced convex wedge with
apex $p$, or there exists both a negative convex wedge with apex $p$ and a positive convex wedge with
apex $p$. In the latter case we can apply Lemma~\ref{lem:island_graph}.
This point $p$ is given by the following lemma.

\begin{lemma}\textbf{(\cite{ceder, ceder_alg})}\label{lem:ceder}
There exist three lines, concurrent at a point $p$, that divide the plane into six open regions, 
		with the property that every region contains at least $\frac{1}{6}n-1$ points of $S$. 
Moreover, $p$ can be found in $O(n \log n)$ time.
\end{lemma}

The existence of the point $p$ given in Lemma~\ref{lem:ceder} was shown by Ceder~\cite{ceder} 
	using a theorem of Buck and Buck~\cite{buck} on equipartitions of convex sets in the plane. 
The algorithm for finding $p$ is due to Sambuddha and Steiger~\cite{ceder_alg}.
The point $p$ has the property that many of the wedges with apex $p$ and whose bounding rays pass through points of $S$ are convex. 
This is quantified in Lemma~\ref{lem:many_convex}.

\begin{lemma}\label{lem:many_convex}
Let $S$ be a set of $n$ points in the plane, and $p$ a point given by Lemma~\ref{lem:ceder}. 
Let $k<\frac{5}{12}n$ be  a positive integer.
If  $k < \frac{1}{3}n$ then all the wedges with apex $p$ that contain
$k$ points of $S$ and whose bounding rays pass through points of $S$ are convex; 
if $k \ge \frac{1}{3}n$ then at least $2n-3k-3$ of them are convex.
\end{lemma}

\begin{proof}
The six regions in Lemma~\ref{lem:ceder} are all wedges with apex $p$. 
Let $W_0,\dots,W_5$ be these wedges sorted clockwise around $p$. 
Note that since each $W_i$ contains at least $\frac{1}{6}n-1$ points of $S$, any wedge with apex $p$ 
	that contains less than $\frac{1}{3}n$ points is contained in the union of at most three consecutive $W_i$'s.
Thus any such wedge is convex.

Assume that $\frac{1}{3}n \le k < \frac{5}{12}n$ and set $t:=k-\frac{1}{3}n+1$. 
Sort the points of $S$ clockwise by angle around $p$. 
Let $P_i$ and $R_i$ be the first and last $t$ points of $W_i$, respectively, in this order. 
Let $W$ be a wedge with apex $p$, containing $k$ points of $S$ and whose bounding rays pass through points of $S$. 
Note that if the first vertex of $W$ lies in $W_i\setminus{R_i}$ then its last vertex lies in $W_{i+1} \cup W_{i+2}$ (modulo 6).
In this case $W$ is convex. 
Therefore, for $W$ to be non-convex, its first vertex must be in some $R_i$. 
Now, if the first vertex of $W$ is in $R_i$ then its last vertex is in $P_{i+3}$ (modulo 6). 
Let $V_i$ be the set of wedges with apex $p$ that contain $k$ points of $S$, 
	whose bounding rays pass through points of $S$, and whose first vertex is in $R_i$. 
Since $k \le \frac{5}{12}n$, $P_i$ and $R_i$ are disjoint.  
Therefore, all the wedges in $V_i$ are convex or all the wedges in $V_{i+3}$ (modulo 6) are convex. 
This gives $3t$ extra convex wedges.
Summarizing, we have $n-6t$ convex wedges not starting at some $R_i$ plus at least $3t$ extra convex wedges starting at some $R_i$. 
So in total we have at least $2n-3k-3$ convex wedges with apex $p$ that contain $k$ points of $S$.
\end{proof}

\begin{theorem}\label{thm:alg_2}
Let $S$ be a set of $r$ red and $b$ blue points in the plane and 
let $\alpha \in \left [ 0,\frac{1}{2} \right ]$ be such that  
	\[\lceil \alpha r\rceil + \lceil \alpha b \rceil < \frac{1}{3}n+ \frac{2}{3}\sqrt{n+r\lceil \alpha b \rceil-b\lceil \alpha r \rceil}-\frac{4}{3}.\]
Then an island containing exactly $\lceil \alpha r\rceil$ red and exactly $\lceil \alpha b \rceil$ blue points of $S$ can be found in $O( n\log n)$ time.
\end{theorem}

\begin{proof}
Let $k:=\lceil \alpha r\rceil + \lceil \alpha b \rceil$.
Find $p$ as in Lemma~\ref{lem:ceder} in $O( n\log n)$ time. 
We compute the set of wedges, $\mathcal{W}$ that have apex $p$, whose bounding rays pass through points of $S$, and that contain $k$ points of $S$ in $O( n\log n)$ time.
While computing the wedges $\mathcal{W}$, we also compute their weights.
We are done if there is a convex wedge of weight $z:=\lceil \alpha r\rceil/r - \lceil \alpha b\rceil/b$ in $\mathcal{W}$. 
So assume that every convex wedge in $\mathcal{W}$ has weight greater or less than $z$ (that is, has positive or negative weight). 

We show that $\mathcal{W}$ contains a convex negative wedge and a convex positive wedge.
Afterwards, we apply Lemma~\ref{lem:island_graph} to find a balanced island of exactly $\lceil \alpha r\rceil$ red 
	and exactly $\lceil \alpha b \rceil$ blue points of $S$ in  $O( n\log n)$ time.
We only give the proof that $\mathcal{W}$ contains a positive convex wedge. 
The proof that it contains a negative convex wedge is similar. 
 
Let $P$ be the number of non-negative wedges in $\mathcal{W}$ and let $M$  be the sum of the weights of the non-negative wedges.
Note that the sum over all wedges in $\mathcal{W}$ is zero as every point appears in the same number of wedges. 
Therefore, the sum of the weights of the negative wedges is equal to $-M$. 
In particular, this implies that there exist both negative and positive wedges.  
If $k < \frac{1}{3} n$ then by Lemma~\ref{lem:many_convex} all the wedges in $\mathcal{W}$ are convex and we are done.  
Assume that $k \ge \frac{1}{3} n$.

Note that the weight of two consecutive wedges in $\mathcal{W}$ differs by at most $\delta:=\frac{1}{r}+\frac{1}{b}$. 
For $M$ to achieve its largest possible value, the non-negative wedges must lie consecutively as follows.  
The first wedge has weight $z$.
Subsequent wedges increase in weight by $\delta$ until they reach a maximum. 
Afterwards, subsequent wedges decrease in weight by $\delta$ until they reach $z$ again. 
Depending on whether $P$ is even or odd the sequence will stay at this maximum value for one or two wedges of the sequence. 
In both cases, simple arithmetic shows that $M \le \frac{1}{4}\delta(P^2-2P+1)+Pz$. 

The largest possible negative weight is $z-\delta$. 
Therefore the sum of the negative weights is at most $(z-\delta)(n-P)$. 
Thus, $M \ge  (\delta -z)(n-P)$, and 
		\[\frac{1}{4}\delta(P^2-2P+1)+Pz-(\delta-z)(n-P)\ge 0.\] 
Solving for $P$, we have that  
		\[P \ge 2\sqrt{n-\frac{nz}{\delta}}-1.\] 
Given that $n=r+b$, $\delta=\frac{1}{r}+\frac{1}{b}$ and $z=\lceil \alpha r\rceil/r - \lceil \alpha b\rceil/b$, this implies that 
	\[ P \ge 2 \sqrt{n+r\lceil \alpha b \rceil-b\lceil \alpha r \rceil}-1.\] 
By Lemma~\ref{lem:many_convex}, at least $2 n-3k-3$ of the wedges of $\mathcal{W}$ are convex. 
Thus, since $k < \frac{1}{3} n+ \frac{2}{3}\sqrt{n+r\lceil \alpha b \rceil-b\lceil \alpha r \rceil}-\frac{4}{3}$, 
	the number of non-convex wedges of $\mathcal{W}$ is at most $3k+3- n<2\sqrt{n-b\lceil \alpha r \rceil-r\lceil \alpha b \rceil}-1 \le P$.  
Therefore, at least one of the convex wedges of $\mathcal{W}$ is non-negative.
Since all balanced wedges are non-convex by assumption, this wedge must be positive and the result follows.
\end{proof}

\subsection*{Acknowledgments}

We thank Pablo Sober\'on for pointing out that the Theorem of Blagojevi\'c and Dimitrijevi\'c Blagojevi\'c
implies the first case of Theorem~\ref{thm:hobby}.

\small 
\bibliographystyle{abbrv}
\bibliography{balancedbib}

\end{document}